\documentclass[conference, letterpaper]{IEEEtran}
\IEEEoverridecommandlockouts
\usepackage{cite}
\usepackage{amsmath,amssymb,amsfonts}
\usepackage{amsthm}
\usepackage{algorithmic}
\usepackage[ruled,vlined, linesnumbered]{algorithm2e}
\usepackage{graphicx}
\usepackage{textcomp}
\usepackage{enumerate}
\usepackage{tikz}
\usepackage{xcolor}
\def\BibTeX{{\rm B\kern-.05em{\sc i\kern-.025em b}\kern-.08em
    T\kern-.1667em\lower.7ex\hbox{E}\kern-.125emX}}
    
\usepackage{float}
\usepackage{color}
\usepackage{balance}
\newcommand{\R}{\mathbb{R}}

\newtheorem{definition}{Definition}
\newtheorem{theorem}{Theorem}
\newtheorem{lemma}{Lemma}

\newtheorem{example}{Example}
\newtheorem{remark}{Remark}

\definecolor{Rafael}{RGB}{139,0,0}

\definecolor{rick}{RGB}{0,139,0}

\definecolor{old}{RGB}{0,0,139}

\definecolor{Benedikt}{RGB}{50, 80, 130}

\begin{document}
\title{A Reverse Jensen Inequality Result with Application to Mutual Information Estimation}
\author{
	  \IEEEauthorblockN{Gerhard Wunder$^{\ast}$, Benedikt Gro\ss$^{\ast}$, Rick Fritschek$^{\ast}$, and Rafael F. Schaefer$^{\dagger}$\\[2mm]}
	  \IEEEauthorblockA{\small
	    \begin{tabular}{cc}
	       \begin{tabular}{c}
	           $^{\ast}$ Cybersecurity and AI Group\\
                        Freie Universit\"at Berlin \\
                  \small      Takustr. 9, 14195 Berlin, Germany\\
                        \texttt{\footnotesize\{g.wunder, benedikt.gross, rick.fritschek\}@fu-berlin.de}
	       \end{tabular}
	       \begin{tabular}{c}
	           $^{\dagger}$ Chair of Communications Engineering and Security\\
	                        University of Siegen \\
               \small             H\"olderlinstr. 3, 57068 Siegen, Germany\\
                            \texttt{\footnotesize rafael.schaefer@uni-siegen.de}
	       \end{tabular}
	  \end{tabular}
}
\thanks{The work of G. Wunder and B. Gro\ss{} was supported by the German Research Foundation (DFG) under Grants DFG WU 591/7-1, WU 591/8-2, and WU 598/11-1. The work R. Fritschek was supported by the DFG under Grant FR 4209/1-1. The work of R. F. Schaefer was supported by the DFG under Grant SCHA 1944/7-1.}}

\maketitle

\begin{abstract}
The Jensen inequality is a widely used tool in a multitude of fields, such as for example information theory and machine learning. 
It can be also used to derive other standard inequalities such as the inequality of arithmetic and geometric means or the Hölder inequality. In a probabilistic setting, the Jensen inequality describes the relationship between a convex function and the expected value. In this work, we want to look at the probabilistic setting from the reverse direction of the inequality. We show that under minimal constraints and with a proper scaling, the Jensen inequality can be reversed. We believe that the resulting tool can be helpful for many applications and provide a variational estimation of mutual information, where the reverse inequality leads to a new estimator with superior training behavior compared to current estimators.
\end{abstract}

\section{Introduction}
The Jensen inequality is an important tool in many fields of pure and applied mathematics, such as convex analysis, information theory, and machine learning.
For an integrable, real-valued random variable $X$ and a convex function $\varphi$ it states that
\begin{equation}
    \label{eq:jensen}
    \varphi\left(\mathbb{E}[X]\right)\leq \mathbb{E}\left[\varphi(X)\right]. 
\end{equation}
It can be used to derive other standard inequalities such as the inequality of arithmetic and geometric means or the Hölder inequality.
In this paper, we proof a reversion of Jensen's inequality for a certain class of measures. Several reversion results for specific functions or classes of distributions are known in the literature. In \cite{jebara2001reversing}, a reverse Jensen inequality for a family of exponential distributions is used to obtain tractable bounds for expectation maximisation (EM) for conditional and discriminative learning of certain latent variable models. In \cite{budimir2000further}, several reverse Jensen results for discrete variables are formulated with applications in information theory. Further results can be found for example in \cite{dragomir1994some, dragomir2013some}.
In \cite{khan2020converses}, a converse of the Jensen inequality is derived and 
used to compute upper bounds on some divergences.

This paper is structured as follows. In Section~\ref{sec:reversed_jensen}, we derive our reverse inequalities and give concrete examples for its application.
In Section \ref{sec:mi_estimation}, the reverse inequality is used to obtain a variational lower bound on the mutual information (MI). Mutual information is a central quantity in information theory that measures the amount of information one random variable carries about another. Recently, MI estimators gained much interest in the field of machine learning, e.g., in the context of representation learning \cite{hjelm2018learning}, generative models and deep information bottleneck but also in wireless communications \cite{8815464,fritschek2020deep, fritschek2020neural}. In Section~\ref{sec:mi_estimation} an overview over recent developments in the use of MI for several machine learning applications and estimators of MI from finite samples is given. Section~\ref{sec:numerics} validates our MI estimator on synthetic data. The paper concludes with a discussion on limitations of MI estimation from finite samples and future research directions in Section~\ref{sec:conclusion}.

\section{Reverse Jensen Inequalities}
\label{sec:reversed_jensen}

\subsection{Concave Functions}

Our initial lemma is as follows. Let
$b=b\left( p\right) :=\frac{\mathbb{E}\left[ {X^{p}}\right] }{\mathbb{E}\left[
{X}\right] ^{p}}\geq 1$
be the ratio of first and $p$-th non-centralized moment.

\begin{lemma}
\label{lemma1}
Let $f:\mathbb{R}_{0}^{+}\rightarrow \mathbb{R}_{0}^{+}$ be concave and $%
f(0)=0$ and set%
\begin{equation*}
\zeta _{b}\left( a\right) :=\sup_{\frac{1}{p}+\frac{1}{q}=1}\frac{\left[
1-b(p)^{\frac{1}{p}}a^{-\frac{1}{q}}\right] ^{+}}{a}.
\end{equation*}
Then for any random variable $X$, we have
\begin{equation*}
f\left( {\mathbb{E}\left[ {X}\right] }\right) {\geq }\mathbb{E}\left[ {f(X)}%
\right] \geq \sup_{a>b}f(a\mathbb{E}\left[ {X}\right] )\zeta _{b}\left(
a\right). 
\end{equation*}
\end{lemma}

\begin{remark}
The case $b=$ $\infty $ with vanishing first or second moment yields a
trivial lower bound.
\end{remark}

\begin{remark}
It can be easily checked that in the deterministic case where $b\left(
p\right) =1$ and $\zeta _{b}\left( a\right) =\max_{a>1}(1-a^{-1})/a$. In all cases, the inequality appears to be not tight.
\end{remark}

\begin{proof}
The left inequality is the standard Jensen inequality. For any $X^{\ast }>0$
we have 
\begin{align*}
\mathbb{E}\left[ {f(X)}\right] & \geq \mathbb{E}\left[ {f(X)\mathbb{I}%
\left\{ X\leq X^{\ast }\right\} }\right]  \\
& =\mathbb{E}\left[ {\frac{f(X)}{X}X\mathbb{I}\left\{ X\leq X^{\ast
}\right\} }\right]  \\
& \geq \frac{f(X^{\ast })}{X^{\ast }}\mathbb{E}\left[ {X\mathbb{I}\left\{
X\leq X^{\ast }\right\} }\right]  \\
& =\frac{f(X^{\ast })}{X^{\ast }}\mathbb{E}\left[ {X(1-\mathbb{I}\left\{ X
>X^{\ast }\right\} )}\right]  \\
& =\frac{f(X^{\ast })}{X^{\ast }}\left( \mathbb{E}\left[ {X}\right] -\mathbb{%
E}\left[ {X\mathbb{I}\left\{ X >X^{\ast }\right\} }\right] \right),
\end{align*}%
where first inequality follows from the positiveness of $X$, the second from
concavity and the intercept theorem. For $0<p,q<\infty $ with $%
p^{-1}+q^{-1}=1$, we have:%
\begin{align*}
\mathbb{E}\left[ {f(X)}\right] & \geq \frac{f(X^{\ast })}{X^{\ast }}\left( 
\mathbb{E}\left[ {X}\right] -\mathbb{E}^{\frac{1}{p}}\left[ {X^{p}}\right] 
\mathbb{E}^{\frac{1}{q}}\left[ {\mathbb{I}\left\{ X >X^{\ast }\right\} }%
\right] \right)  \\
& =\frac{f(X^{\ast })}{X^{\ast }}\mathbb{E}\left[ {X}\right] \left( 1-\left( 
\frac{\mathbb{E}\left[ {X^{p}}\right] }{\mathbb{E}^{p}\left[ {X}\right] }%
\right) ^{\frac{1}{p}}\Pr (X >X^{\ast })^{\frac{1}{q}}\!\right) \\
& \geq \frac{f(X^{\ast })}{X^{\ast }}\mathbb{E}\left[ {X}\right] \left(
1-\left( \frac{\mathbb{E}\left[ {X^{p}}\right] }{\mathbb{E}^{p-\frac{p}{q}}
\left[ {X}\right] (X^{\ast })^{\frac{p}{q}}}\right) ^{\frac{1}{p}}\right) ,
\end{align*}%
where we applied H\"{o}lder's inequality in the first step and Markov's
inequality in the second step.
If we choose $X^{\ast}=a\mathbb{E}\left[
{X}\right]  $ we get
\[
\mathbb{E}\left[  {f(X)}\right]  \geq\frac{f(a\mathbb{E}\left[  {X}\right]
)}{a}\left(  1-b^{\frac{1}{p}}\left(  p\right)  a^{\frac{1}{q}}\right)
\]
and since the inequality is true for any sequence $p_{n}$, we obtain%
\begin{align*}
&  \mathbb{E}\left[  {f(X)}\right]  \geq\underset{n\rightarrow\infty}{\lim\sup}\frac{f(a\mathbb{E}\left[
{X}\right]  )}{a}\left(  1-b^{\frac{1}{p_{n}}}\left(  p\right)  a^{\frac
{1}{p_{n}}-1}\right)
\end{align*}
and the claim follows.
\end{proof}

This recovers the corresponding lemma in \cite{SWJ12_TWC} in case of $p=q=2$. Obviously, a
non-trivial lower bound exists if
\begin{equation}
\label{eq:moments}
\frac{\mathbb{E}\left[ {X^{p}}\right] }{\mathbb{E}^{p}\left[ {X}\right] }%
<\infty \text{ for some }p\geq 1.
\end{equation}%
Its actual tightness depends on the growth of the higher moments. 
%
%
%
Condition (\ref{eq:moments}) is certainly true for any positive random variable with finite support, i.e.:%
\begin{equation}
X=0\text{ for }X>a\mathbb{E}\left[  {X}\right]  \label{eq:finite}%
\end{equation}
by the Lebesgue space inclusion theorem. Then, Lemma \ref{lemma1} can be strenghted.

\begin{lemma}
\label{lemma2}
Let $f:\mathbb{R}_{0}^{+}\rightarrow\mathbb{R}_{0}^{+}$
be concave and $f(0)=0$. Then, for any random variable fulfilling condition
(\ref{eq:finite})%
\[
f\left(  {\mathbb{E}\left[  {X}\right]  }\right)  {\geq}\mathbb{E}\left[
{f}\left(  {X}\right)  \right]  \geq\frac{f\left(  a\mathbb{E}\left[
{X}\right]  \right)  }{a}.%
\]
The inequality is tight for the discrete measure:%
\[
X=\left\{
\begin{array}
[c]{cc}%
0 & \text{w.p. }\frac{m-1}{m}\\
m & \text{w.p. }\frac{1}{m}%
\end{array}
\right.
\]
\begin{proof}
Following the same lines as in Lemma \ref{lemma1}, we have%
\begin{equation*}
\mathbb{E}\left[ {f(X)}\right] \geq \frac{f(X^{\ast })}{X^{\ast }}\mathbb{E}%
\left[ {X}\right] \left( 1\!-\!\left( \frac{\mathbb{E}\left[ {X^{p}}\right] }{%
\mathbb{E}^{p}\left[ {X}\right] }\right) ^{\frac{1}{p}}\Pr (X \!>\!X^{\ast })^{%
\frac{1}{q}}\right) 
\end{equation*}
Setting again $X^{\ast}=a\mathbb{E}\left[
{X}\right]  $ and using the condition (\ref{eq:finite}) yields the result.
\end{proof}
\end{lemma}

\begin{example}
We can use this result for the special case
\begin{equation*}
g{(z_{1},...,z_{n})=}\sum\limits_{i=1}^{n}\log \left( 1+z_{i}\right)
,\;z_{i}\geq 0.
\end{equation*}%
Setting ${f(z)=\log (1+z)}$ and letting $Z$ be a discrete random variable in $%
\mathbb{R}_{0}^{+}$ taking on values $z_{1},...,z_{n}$ each with
probabilities $p_{Z}\left( z_{i}\right) =1/n$
the above result yields%
\begin{align*}
\log \left( 1+\frac{1}{n}\sum\limits_{i=1}^{n}z_{i}\right) &\geq \frac{1}{n%
}{\sum\limits_{i=1}^{n}\log \left( 1+z_{i}\right)  }\\
&\geq\sup_{a>b}\log
\left( 1+\frac{a}{n}\sum\limits_{i=1}^{n}z_{i}\right) \zeta _{b}\left(
a\right) .
\end{align*}
\end{example}

Another often used example includes ${f(z)=}\frac{z}{1+z}$ where similar
inequalities can be obtained.

\subsection{Convex Functions}
We have the following lemma for convex functions.



\begin{lemma}
\label{lemma3}
Let $f:\mathbb{R}_{0}^{+}\rightarrow\mathbb{R}_{0}^{+}$ be convex, increasing
and $f(0)=0$. Then for any random variable $X\geq0$%
\[
f\left(  \zeta_{b}^{-1}\left(  a\right)  \mathbb{E}\left[  {X}\right]
\right)  \geq a\mathbb{E}\left[  {f}\left(  X\right)  \right] .
\]
\end{lemma}

\begin{proof}
We have from Lemma \ref{lemma1} with concave $g$
\[
\mathbb{E}_{{X}}\left[  g{(X)}\right]  \geq g(a\mathbb{E}_{{X}}\left[
{X}\right]  )\zeta_{b}\left(  a\right)
\]
where we indicate with $\mathbb{E}_{{X}}$ that the expectation is with respect
to the distribution of $X$. Setting $Y={g(X)}$ and using change of measures in the
expectation yields:%
\[
\mathbb{E}_{{Y}}\left[  {Y}\right]  \geq g(a\mathbb{E}_{{Y}}\left[  {g}%
^{-1}\left(  Y\right)  \right]  )\zeta_{b}\left(  a\right)
\]
and eventually:%
\[
g^{-1}\left(  \zeta_{b}^{-1}\left(  a\right)  \mathbb{E}_{{Y}}\left[
{Y}\right]  \right)  \geq a\mathbb{E}_{{Y}}\left[  {g}^{-1}\left(  Y\right)
\right]
\]
Since $X$ and thus $Y$ and $g$ were arbitrary (under the posed constraints) the result follows by replacing the convex
$g^{-1}$ with $f$.
\end{proof}

\begin{remark}
The result gives a nice moment bound. Let $X=0$ for $X>a\mathbb{E}_{{X}}[X]$ and $f\left(  X\right)  =Y=X^{\frac{1}{c}}$ then
$f^{-1}\left(  Y\right)  =X=Y^{c}$ for some $c>1$. By Lemma \ref{lemma2}, i.e. $\zeta_{b}^{-1}(a)=a$, and Lemma \ref{lemma3}, i.e. $a^{c}\mathbb{E}_{{Y}}^{c}\left[  {Y}\right]  \geq a\mathbb{E}_{{Y}}\left[
{Y}^{c}\right]$, we have%
\[
\mathbb{E}_{{Y}}\left[  {Y}^{c}\right]  \leq a^{c-1}\mathbb{E}_{{Y}}%
^{c}\left[  {Y}\right] .
\]

\end{remark}

A result for decreasing $f$ is as follows.
\begin{lemma}
\label{lemma4}
Suppose $f:\mathbb{R}_{0}^{+}\rightarrow\mathbb{R}_{0}^{+}$ is some decreasing
convex function with $f\left(  x\right)  =c$. Then, we have for any random
variable $Z\geq x$
\begin{align*}
& \mathbb{E}\left[  {f(Z+x)}\right] \\ & \leq\inf_{a>b}\left\{  f{(\mathbb{E}%
\left[  a{Z+x}\right]  )}+\left[  c-f{(\mathbb{E}\left[  a{Z+x}\right]
)}\right]  \left(  1-\zeta_{b}\left(  a\right)  \right)  \right\} .
\end{align*}

\end{lemma}

\begin{proof}
The trick is to transform the problem into an equivalent concave problem. Set
$F\left(  z\right)  :=-f\left(  z+x\right)  +c$ so that $F$ is concave and
$F\left(  0\right)  =0,F\left(  z\right)  \geq0$. By Lemma \ref{lemma1}, we
have $\mathbb{E}\left[  F{(Z)}\right]  \geq F(a\mathbb{E}\left[
{Z}\right]  )\zeta_{b}\left(  a\right)  $ for $a>b$. Substituting this, we get%
\[
\mathbb{E}\left[  -f{(Z+x)}+c\right]  \geq\left(  -f{(a\mathbb{E}\left[
{Z}\right]  +x)}+c\right)  \zeta_{b}\left(  a\right)  >0
\]
so that
\begin{align*}
-\mathbb{E}\left[  f\left(  {Z+x}\right)  \right]  
&  \geq-f\left(  {\mathbb{E}\left[  aZ{+x}\right]  }\right)  \zeta_{b}\left(
a\right)  +c\zeta_{b}\left(  a\right)  -c.
\end{align*}
Multiplying by $-1$ yields the result. %
\end{proof}

\begin{example}
Let $f$ be the function $f\left(  z\right)  :=1/z$ and let $Z$ be a discrete
random variable with ${\mathbb{E}\left[  {Z}\right]  =:}\bar{z}$ in
$\mathbb{R}_{0}^{+}$ taking on $z_{1},...,z_{n}$ with
$p_{Z}\left(  z_{i}\right)  =1/n$. Then, we have by Lemma \ref{lemma4}%
\[
\frac{1}{\bar{z}}\leq\frac{1}{n}\sum\limits_{i}\frac{{1}}{{z}_{i}}\leq
\frac{{1}}{a\left(  \bar{z}-z_{1}\right)  +z_{1}}+\frac{\left(  1-\zeta
_{b}\left(  a\right)  \right)  }{z_{1}}.%
\]

\end{example}

\subsection{Special Functions}

\subsubsection{The log-function}

The usefulness is shown in the next theorem.

\begin{theorem}
\label{thm:log} For any random variable $X>0$ with $\zeta _{b}\left(
a\right) ,b$ (defined above), any $a>b$, we have the lower and upper bound
\begin{align*}
 \log \left( \mathbb{E}\left[ X\right] \right)  &\leq \frac{1}{\zeta _{b}\left( a\right) }\mathbb{E}\left[ \log \left( 
\frac{{1+X}}{\left( \frac{1}{\mathbb{E}\left[ {X}\right] }+a\right)
^{^{\zeta _{b}\left( a\right) }}}\right) \right]  \\
& \leq \frac{1}{\zeta _{b}\left( a\right) }\log \left( 1+\mathbb{E}\left[ X%
\right] \right) +\log \left( \frac{\mathbb{E}\left[ X\right] }{1+a\mathbb{E}%
\left[ X\right] }\right) .
\end{align*}%
of which the difference
is uniformly bounded in 
$\mathbb{E}\left[ X\right] $.
\end{theorem}

\begin{proof}
For the left-hand side, fix some $c>0,a>b$, then%
\begin{align*}
& \mathbb{E}\left[ \log \left( \frac{{1+X}}{c}\right) \right] \\
&\quad \geq \zeta _{b}\left( a\right) \log \left( 1+a\mathbb{E}\left[ {X}\right]
\right) +\log \left( \frac{{1}}{c}\right)  \\
&\quad =\zeta _{b}\left( a\right) \log \left( \frac{1+a\mathbb{E}\left[ {X}\right]
}{c^{\frac{1}{\zeta _{b}\left( a\right) }}}\right) 
\end{align*}%
using the above lemma. For any (appropriate) $c\left( a,\mathbb{E}\left[ {X}%
\right] \right) >0$ we obtain%
\begin{align*}
\log \left( \mathbb{E}\left[ X\right] \right) & \leq \log \left( \frac{1\!+\!a%
\mathbb{E}\left[ {X}\right] }{c^{\frac{1}{\zeta _{b}\left( a\right) }}}%
\right)  
 \leq \frac{1}{\zeta _{b}\left( a\right) }\mathbb{E}\left[ \log \left( 
\frac{{1\!+\!X}}{c}\right) \right] .
\end{align*}%
Selecting the optimal parameter $c\left( a,\mathbb{E}\left[ {X}\right]
\right) $ as%
\begin{equation*}
c\left( a,\mathbb{E}\left[ {X}\right] \right) =\left( \frac{1}{\mathbb{E}%
\left[ {X}\right] }+a\right) ^{\zeta _{b}\left( a\right) }
\end{equation*}%
yields the desired result.

The right-hand side is%
\begin{align*}
& \frac{1}{\zeta _{b}\left( a\right) }\mathbb{E}\left[ \log \left( \frac{{1+X%
}}{\left( \frac{1}{\mathbb{E}\left[ {X}\right] }+a\right) ^{\zeta _{b}\left(
a\right) }}\right) \right]  \\
&\quad =\frac{1}{\zeta _{b}\left( a\right) }\mathbb{E}\left[ \log \left( {1+X}%
\right) \right] +\mathbb{E}\left[ \log \left( \frac{\mathbb{E}\left[ {X}%
\right] }{1+\mathbb{E}\left[ {X}\right] }\right) \right]  \\
&\quad \leq \frac{1}{\zeta _{b}\left( a\right) }\log \left( {1+}\mathbb{E}\left[ {%
X}\right] \right) +\log \left( \frac{\mathbb{E}\left[ {X}\right] }{1+\mathbb{%
E}\left[ {X}\right] }\right) 
\end{align*}%
due to Jensen's inequality. Eventually, it is easy to see that the difference is uniformly bounded.
\end{proof}


\subsubsection{Product-type functions}

Suppose we have real numbers $a_{1},...,a_{n}\geq \alpha $ and $%
b_{1},...,b_{n}\geq \beta>0 $ (denominator shall be non-zero). Set without loss
of generality $\alpha =\beta =1$. Further, suppose $a_{i}\left( \phi \right) $
and $b_{i}\left( \phi \right) $ both depend on some parameter set $\phi
\subseteq D$ such that any $\phi \subseteq D$ guarantees the lower bounds
for $a_{i},b_{i}$. Consider the metric to be maximized over $D$ as 
\begin{equation*}
f\left( \phi \right) =\prod\limits_{i=1}^{n}\frac{a_{i}\left( \phi \right) 
}{b_{i}\left( \phi \right) }.
\end{equation*}%
Multiply by $\beta /\alpha $ if the lower does not hold. Let $C\left(
s_{a}\right) $ and $C\left( s_{b}\right) $ be the \emph{spread functions}
that depend on the spreads%
\begin{align*}
s_{a}&:=\frac{\max_{i}a_{i}}{\frac{1}{n}\sum\nolimits_{i=1}^{n}a_{i}}\geq 1
,\qquad
s_{b} :=\frac{\max_{i}b_{i}}{\frac{1}{n}\sum\nolimits_{i=1}^{n}b_{i}}\geq 1
\end{align*}%
and 
\begin{align*}
C\left( s_{a/b}\right) & :=\max_{c_{1}}\frac{1}{c_{1}}\left( 1-\sqrt{\frac{%
1+s_{a/b}}{c_{1}}}\right)  .
\end{align*}
Then using above lemmas we can prove the following theorem.

\begin{theorem}
Suppose we have real numbers $a_{i}\left( \phi \right) \geq 1$ and $%
b_{i}\left( \phi \right) \geq 1$ depending on some parameter set $\phi
\subseteq D$. We have 
\begin{equation*}
\frac{\left( \frac{1}{n}\sum\nolimits_{i}a_{i}\left( \phi \right) \right)
^{C\left( s_{a}\right) }}{\frac{1}{n}\sum\nolimits_{i}b_{i}\left( \phi
\right) }\leq \left( \prod\limits_{i=1}^{n}\frac{a_{i}\left( \phi \right) }{%
b_{i}\left( \phi \right) }\right) ^{\frac{1}{n}}\leq \frac{\frac{1}{n}%
\sum\nolimits_{i}a_{i}\left( \phi \right) }{\left( \frac{1}{n}%
\sum\nolimits_{i}b_{i}\left( \phi \right) \right) ^{C\left( s_{b}\right) }}
\end{equation*}%
where $C\left( s_{a}\right) ,C\left( s_{b}\right) $ are the spread
parameters.
\end{theorem}

Obviously, depending on the spread of $a_{1},...,a_{n}$ and $b_{1},...,b_{n}$
we can replace product terms with sums. Applications of such estimations can
be applied in waveform design, see e.g. \cite{SPW+17_PIMRC}.

The proof is omitted due to lack of space.

\section{Neural Estimation of Mutual Information}
\label{sec:mi_estimation}
Precise estimation of mutual information is a long standing problem due to its dependence on the underlying probabilities, which are normally unknown in most applications, and only observable through samples. Classical approaches tackled this problem based on a binning of the probability
space \cite{fraser1986independent,darbellay1999estimation}, $k$-nearest
neighbor statistics
\cite{kraskov2004estimating,gao2015efficient,GaoEstimatorMI}, maximum
likelihood estimation \cite{suzuki2008approximating}, and variational lower
bounds \cite{barber2003algorithm}. However, most of these techniques are limited to very low dimensional problems. However, due to its use in the field of deep learning, for example as a metric in representation learning, it received renewed interest. In particular, \cite{belghazi2018mine}, showed a promising new direction in combining deep learning methods with variational bounds which resulted in a surge of papers building upon this new direction. It also presented the first estimator within this framework, the mutual information neural estimator (MINE), which is based on the Donsker-Varadhan lower bound of the Kullback-Leibler divergence but is biased, see~\cite{poole2018variational}. Another closely related estimator is the Nguyen-Wainwright-Jordan (NWJ) estimator, based on the lower bound \cite{nguyen2010estimating}, also known as f-GAN \cite{nowozin2016f}. This estimator is unbiased and can be derived through the Fenchel duality. Both estimators are limited by their variance which scales exponentially with the mutual information. The info noise contrastive estimation (InfoNCE) estimator \cite{oord2018representation} was derived in the context of representation learning, and gives a bound with low variance, but high bias, dependent on the sample size and the mutual information. Recent works have combined the high variance, low bias estimators (NWJ, MINE) with the high bias, low variance NCE method in \cite{poole2018variational}, \cite{sinha2021hybrid}. Note that the marginal term is the culprit of high variance fluctuations due to the exponential function inside the expected value. This part of the estimator can be dominated by extremely rare events, which are unlikely to be sampled from, see~\cite{pmlr-v108-mcallester20a}.

\subsection{Mutual Information Estimators from Unnormalized Gibbs Measures}
Let us identify $P(X,Y)$ with the ground
truth joint probability measure $\mathbb{P}$ and $\mathbb{Q}:=\mathbb{P}_{X}%
\times\mathbb{P}_{Y}$ with the product of the marginals $P(X)$ and $P(Y)$. We use the notation $\mathbb{P}\ll\mathbb{Q}$, to denote that $\mathbb{P}$ is absolutely continuous w.r.t. $\mathbb{Q}$.
Further, let $\mathbb{P}^{n}$ and $\mathbb{Q}^{n}$ denote the empirical
measures from a set of i.i.d. samples. 
Moreover, let $\mathcal{F}$ be a family of functions $T_{\theta}:
\mathcal{X} \times\mathcal{Y}\rightarrow\mathbb{R}$ parametrized by the
weights $\theta\in\Theta$ which can be a neural network.

To derive the Donsker-Varadhan representation of the Kullback-Leibler divergence, i.e., the starting point of MINE
\begin{equation*}
   D_{KL}(\mathbb{P}||\mathbb{Q})=\sup_{T:\Omega\rightarrow \mathbb{R}} \mathbb{E}_{\mathbb{P}}[T]-\log \mathbb{E}_{\mathbb{Q}}[e^T] ,
\end{equation*} one can utilize Gibbs measures. Let $\mathbb{G}$ be any positive
measure with total variation $|\mathbb{G}|=\int d\mathbb{G}$. One can now identify a Gibb's measure $d\mathbb{G}=\frac{e^{T}}%
{\mathbb{E}_{\mathbb{Q}}[e^{T}]}d\mathbb{Q}$ for a given critic $T$, which needs to be optimized. One can see that by construction,
$\mathbb{Q}$ dominates $\mathbb{G}$, i.e.,$\mathbb{G}\ll \mathbb{Q}$ and $\mathbb{G}$ is a probability measure.

Hence, we have%
\begin{equation*}
\mathbb{E}_{\mathbb{P}}\log\frac{d\mathbb{G}}{d\mathbb{Q}} =\mathbb{E}%
_{\mathbb{P}}[T]-\log\mathbb{E}_{\mathbb{Q}}[e^{T}] \leq
\mathbb{E}_{\mathbb{P}}\log\frac{d\mathbb{P}}{d\mathbb{Q}} \label{eq:ineq}%
\end{equation*}
where the inequality is due to the
positiveness of the Kullback-Leibler divergence $D_{\text{KL}}(\mathbb{P}%
\|\mathbb{G})$, cf. \cite[Th. 4]{belghazi2018mine}. Moreover, a non-unique optimum is $T^{\ast}=\log
\frac{d\mathbb{P}}{d\mathbb{Q}}+c$, i.e., the estimate of MINE can be
unnormalized. However, as briefly discussed above, the MINE estimator suffers from a bias because of the $\log$ in the marginal term, when using Monte Carlo sampling, i.e.
$\mathbb{E}_{\mathbb{Q}}[\log\mathbb{E}_{\mathbb{Q}^{n}}[e^{T}]]\leq
\log\mathbb{E}_{\mathbb{Q}}[e^{T}]$, which leads to instability.

An interesting new direction is obtained when we consider the family of unnormalized Gibbs measures.
\begin{definition}
An unnormalized Gibbs measure is defined through%
\[
d\mathbb{G}=\frac{e^{T}}{Z\left( T\right)  }d\mathbb{Q},\;\exists
c\in\mathbb{R}:\int\frac{e^{T+c}}{Z\left(  T+c\right)  }d\mathbb{Q}=1
\]
where $Z$ is some normalization function.
\end{definition}

For the NWJ estimator, where $d\mathbb{G}=\frac{e^{T}}{\exp\left(  e^{-1}\mathbb{E}_{\mathbb{Q}%
}[e^{T}]\right)  }d\mathbb{Q}$,
we have that%
\begin{align*}
\mathbb{E}_{\mathbb{P}}\log\frac{d\mathbb{G}}{d\mathbb{Q}}  &  =\mathbb{E}%
_{\mathbb{P}}[ T  ] -e^{-1}\mathbb{E}_{\mathbb{Q}%
}[e^{T}]\\
&  \leq\mathbb{E}_{\mathbb{P}}[T]-\log\mathbb{E}_{\mathbb{Q}}[e^{T}]
\leq\mathbb{E}_{\mathbb{P}}\log\frac{d\mathbb{P}}{d\mathbb{Q}}%
\end{align*}
by the simple inequality $\log\left(  x\right)  \leq\frac{x}{e}$, cf. \cite{poole2018variational}. Notably, the
measures identified with $T$ in the first and second line are actually
different, but all inequalities become tight for $T^{\ast}=\log\frac
{d\mathbb{P}}{d\mathbb{Q}}+1$. As discussed earlier, the resulting NWJ estimator is unbiased, due to $\mathbb{E}_{\mathbb{Q}}[e^{-1}\mathbb{E}_{\mathbb{Q}^{n}}[e^{T}]]=
\mathbb{E}_{\mathbb{Q}}[e^{T-1}]$.

\subsection{Bounds on Variance}
Lower bounds for the marginal terms of the MINE and NWJ estimator follow from the analysis in 
\cite[Th. 2]{song2019understanding}. For MINE, the variance $\mathbb{V}_{\mathbb{G},\mathbb{Q}}$
can be given as
\[
\liminf_{n}n\mathbb{V}_{\mathbb{G},\mathbb{Q}}[\log\mathbb{E}_{\mathbb{Q}^{n}%
}[e^{T}]]\geq e^{D_{\text{KL}}(\mathbb{G}\|\mathbb{Q})}-1,
\]
while for the NWJ estimator it can be shown that 
\[
\mathbb{V}_{\mathbb{G},\mathbb{Q}}[\mathbb{E}_{\mathbb{Q}^{n}}[e^{T}%
]]\geq\frac{e^{D_{KL}(\mathbb{G}||\mathbb{Q})}-|\mathbb{G}|^{2}}%
{n}%
\]
due to the independence of ${\mathbb{Q}^{n}}$ and
${\mathbb{P}^{n}}$. Thus, it can be seen that the variance of MINE and NWJ scales
exponentially with the estimated mutual information. Notably, \cite{song2019understanding} has addressed this issue with the smoothed mutual information lower bound estimator (SMILE) which simply clips $T$ in the marginal term by $\tau$ and thereby bounds the variance of the marginal term but introduces a bias. However, since the clipping is on the marginal only the estimator is in general instable.

\subsection{Extension with the Reverse Jensen Inequality}
Our new approach is based on the partial converse of Jensen's inequality for the $\log(1+X)$ function above in Theorem~\ref{thm:log}.
With that result, we can now identify with any $T$ the measure
\begin{equation*}
d\mathbb{G}=\max_{a>b}\frac{e^{T}}{\exp \frac{1}{\zeta _{b}\left( a\right)} \mathbb{E}\left[  \log\left(
\frac{1+e^T}{ c} \right)\right]  }d\mathbb{Q},  
\end{equation*}
where $
c\left(  a,\mathbb{E}\left[  e^T\right]  \right)  =\left(  \frac{1}%
{\mathbb{E}\left[  e^T\right]  }+a\right)  ^{\zeta _{b}\left( a\right)}.%
$

Hence, we can write the inequality chain
\begin{align*}
  \mathbb{E}_{\mathbb{P}}\log\frac{d\mathbb{G}}{d\mathbb{Q}}
&  =\mathbb{E}_{\mathbb{P}}\left(  T  \right)  -\min
_{a>b}\left( \frac{1}{\zeta _{b}\left( a\right)} \mathbb{E}\left[ \log \left(
\frac{1+e^T}{ c} \right) \right] \right)  \\
&  \leq\mathbb{E}_{\mathbb{P}}[T]-\log\mathbb{E}_{\mathbb{Q}}%
[e^{T}]  \leq\mathbb{E}_{\mathbb{P}}\log\frac{d\mathbb{P}}{d\mathbb{Q}}.%
\end{align*}

This leads to the following estimator, which is due to its construction, a lower bound on the mutual information:

\begin{definition}
[Reverse-Jensen-Estimator (RJE)]
\begin{equation*}
I_{\text{RJE}}:=\sup_{T\in\mathcal{F}} \mathbb{E}_{\mathbb{P}}^{}[T]-\frac{1}{\zeta _{b}\left( a\right)} \mathbb{E}_{\mathbb{Q}}\left[ \log \left(
\frac{1+e^T}{ c} \right) \right] .
\end{equation*}
\end{definition}
This time (proof deferred) $\frac{1}{\zeta _{b}\left( a\right)} \mathbb{E}_{\mathbb{Q}}\mathbb{E}_{\mathbb{Q}^n}\left[ \log \left(\frac{1+e^T}{ c(a,\mathbb{E}_{\mathbb{Q}^n})} \right) \right]\geq
\log\mathbb{E}_{\mathbb{Q}}[e^{T}]$, provided $\mathbb{E}_{\mathbb{Q}}[e^{T}]$ is large, so that the estimator is stable. This can be also now seen from the numerical results.

\vspace*{0.5\baselineskip}

\section{Numerical Experiments}
\label{sec:numerics}
We implemented two versions of our MI estimator: a) one which contains the moments ratio $b$ defined in Lemma \ref{lemma1}, and b) in which $b$ is treated as a predefined parameter to be set for the problem at hand. Hence, for a batch of $n$ pairs $(x,y)$ drawn from the joint distribution $p_{XY}$ and $n$ pairs $\hat{x}, \hat{y}$ from the product of the marginal distributions $p_Xp_Y$ and a critic $f: \R^m\times \R^m \rightarrow \R$, an estimate of the mutual information $I(X;Y)$ is calculated as follows:

\begin{enumerate}[a)]
\item 
We set the parameter $a$ in Theorem \ref{thm:log} to $a=cr$ for a scalar parameter $c>0$ and moments ratio
\begin{equation*}
    r=\frac{1}{n}\sum\limits_{i=1}^n e^{2f(x_i,y_i)}\left(\frac{1}{n}\sum\limits_{i=1}^n e^{f(\hat{x}_i, \hat{y}_i)}\right)^{-2}.
\end{equation*}
 The MI estimate is then computed as
 \begin{align*}
    \label{eq:est_a}
    \hat{I}(X;Y) &= \frac1n \sum\limits_{i=1}^n f(x_i,y_i)- \log\left( \frac{\frac1n \sum\limits_{i=1}^ne^{f(\hat{x}_i, \hat{y}_i)}}{1-cr\frac1n \sum\limits_{i=1}^ne^{f(\hat{x}_i,\hat{y}_i)}}\right) \\
    &\qquad- \frac{c}{1-\sqrt{1/c}}\frac1n \sum\limits_{i=1}^n\log\left( 1+e^{f(\hat{x}_i,\hat{y}_i)}\right) .
\end{align*}

\item For scalar parameters $0<b<a$, we estimate the MI as
\begin{align*}
    \hat{I}(X;Y) &= \frac1n \sum\limits_{i=1}^n f(x_i,y_i)+
    \log\left( \frac{1}{\frac1n \sum\limits_{i=1}^ne^{f(\hat{x}_i,\hat{y}_i)}}+a\right) \\ &\qquad-\frac{a}{1-\sqrt{b/a}}\frac1n \sum\limits_{i=1}^n \log\left(
    1+e^{f(\hat{x}_i,\hat{y}_i)}\right).
\end{align*}
\end{enumerate}

The performance of the RJE is assessed through a setup with two multivariate Gaussian random variables $X, Y\in\R^n$ with means 0 and covariance matrices $C_X = C_Y = I_n$ and $C_{XY} = \nu I_n$ for $\nu\in (0,1)$.
Hence, 
\begin{equation*}
    \begin{bmatrix}
    X \\
    Y
    \end{bmatrix} \sim \mathcal{N}\left(0, C\right), \quad C= \begin{bmatrix}
    I_n & \nu I_n \\
    \nu I_n & I_n
    \end{bmatrix},  
\end{equation*}
and the mutual information is given by 
   $I(X;Y) = -\frac12 \log\det C$.
The critic $f:\R^n\times \R^n \rightarrow \R$ computes scores from samples drawn from the joint distribution $p_{XY}$, and from the product of marginal distributions $p_Xp_Y$, obtained by pairing the samples of one batch from $X$ with shuffled samples from $Y$, which are then fed into the estimator to produce the MI estimates. The architecture of the critic network is as follows:
The samples from $X$ and $Y$ are concatenated and then forwarded through two fully connected layers with 100 neurons and relu activations. The final layer is a fully connected layer with linear activation and 1 output neuron.
For dimension 50 and varying parameter $\nu$, estimator a) is trained for 400 epochs with a batch size of 256. In the estimator, we set $c=2$. The estimated MI and its variance over 1000 batches is shown in Fig.~\ref{fig:est_a_gaussian} and \ref{fig:est_b_gaussian}. \emph{It is worth noting that the RJE is stable without any clipping whatsoever as suggested.}

\begin{figure}
    \centering
    \includegraphics[width=.45\textwidth]{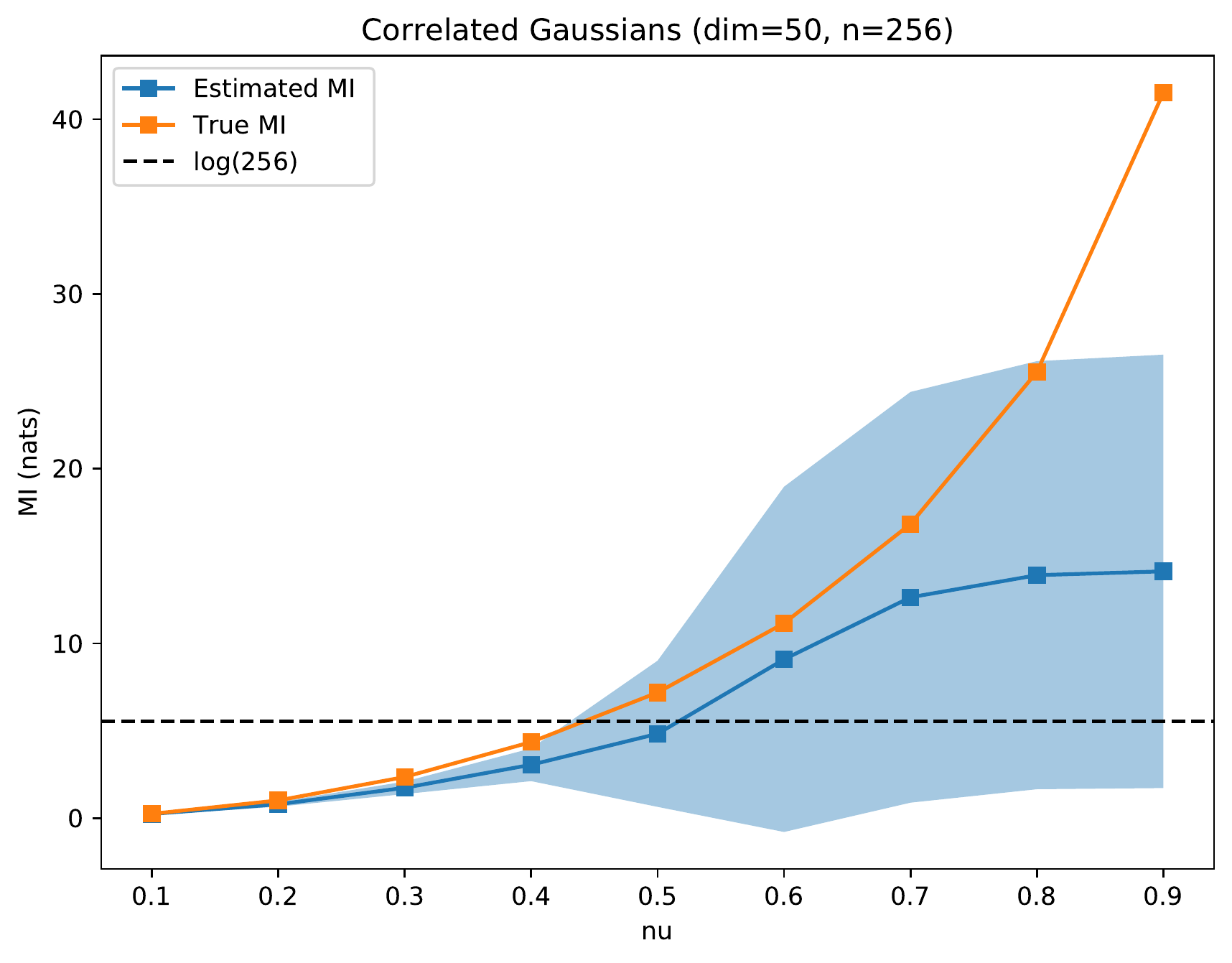}
    \caption{MI estimates for correlated multivariate Gaussians, estimator a)}
    \label{fig:est_a_gaussian}
\end{figure}

\begin{figure}
    \centering
    \includegraphics[width=.45\textwidth]{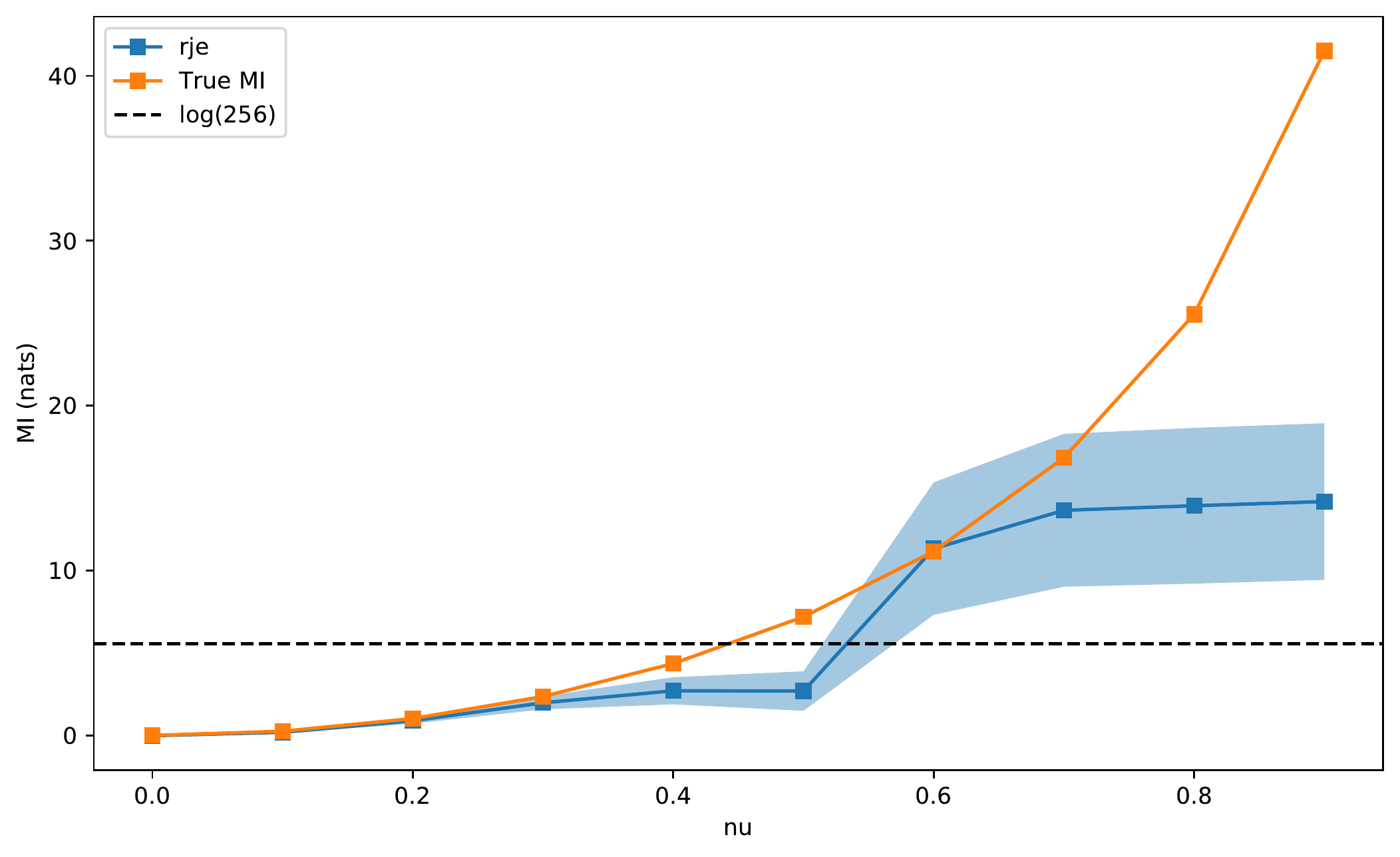}
    \caption{MI estimates for correlated multivariate Gaussians, estimator b)}
    \label{fig:est_b_gaussian}
\end{figure}

\section{Discussion and conclusions}
\label{sec:conclusion}
We proved a reversion of the Jensen inequality for a class of functions with relevance in statistics and machine learning. From these results, we derived an estimator for mutual information from finite samples and evaluated its performance on synthetic data. These types of estimators have gained huge interest in the field of machine learning due to the use of MI as regularization for generative models \cite{chen2016infogan, kingma2013auto} and as objective for representation learning \cite{hjelm2018learning}, among others. Black box estimators of MI such as the one derived here suffer from fundamental limitations \cite{pmlr-v108-mcallester20a}, and many of the estimators found in the literature either fail to capture high MI, or exhibit huge variance, stemming from the marginal term in the variational formulation. By applying our reverse Jensen result, the expectation over the product of marginal distributions on $X$ and $Y$ is pulled out of the logarithm, thereby reducing the variance of the estimation with only a small bias, making the estimator very stable.
The proper tuning of the free parameters in our MI estimators, as well as a thorough investigation of its performance in the machine learning applications mentioned above are a topic of future research.

\medskip

{\small
\newpage
}
\balance
{\small
\bibliographystyle{IEEEtran}
\bibliography{ref}
}

\end{document}